\documentclass[12pt,draftclsnofoot,onecolumn]{IEEEtran}
\usepackage{amsfonts}
\usepackage{graphicx}
\usepackage{color}
\usepackage{amsmath,amsfonts,amssymb,amsthm,epsfig,epstopdf,url,array}
\usepackage{url,textcomp}
\usepackage{authblk}
\usepackage{cite}
\usepackage{caption}
\usepackage{float}
\usepackage[breaklinks,colorlinks,linkcolor=black,citecolor=black,urlcolor=black]{hyperref}

\newtheorem{theorem}{Theorem}
\newtheorem{lemma}{Lemma}

\begin{document}
\title{Performance Analysis of MIMO-HARQ Assisted V2V Communications With Keyhole Effect}
\author{Huan Zhang,
        Zhengtao Liao,
        Zheng Shi,
        Guanghua Yang,
        Qingping Dou,
        and Shaodan Ma
\thanks{\emph{Huan~Zhang and Zhengtao Liao are co-first authors. (Corresponding author: Zheng Shi.)} }  
\thanks{Huan~Zhang and Shaodan~Ma are with the State Key Laboratory of Internet of Things for Smart City and the Department of Electrical and Computer Engineering, University of Macau, Macao 999078, China (e-mails: cquptzh@gmail.com, shaodanma@um.edu.mo).}
\thanks{Zhengtao Liao, Zheng Shi, Guanghua Yang, and Qingping Dou are with the School of Intelligent Systems Science and Engineering, Jinan University, Zhuhai 519070, China (e-mails: {zhengtao@stu2017.jnu.edu.cn}, zhengshi@jnu.edu.cn, ghyang@jnu.edu.cn, tdouqingping@jnu.edu.cn).}  
}
\maketitle
\begin{abstract}
Vehicle-to-vehicle (V2V) communications under dense urban environments usually experience severe keyhole fading effect especially for multi-input multi-output (MIMO) channels, which degrades the capacity and outage performance due to the rank deficiency. To avoid these, the integration of MIMO and hybrid automatic repeat request (HARQ) is proposed to assist V2V communications in this paper. By using the methods of integral transforms, the outage probabilities are derived in closed-form for different HARQ-assisted schemes, including Type I-HARQ, HARQ with chase combining (HARQ-CC), and HARQ with incremental redundancy (HARQ-IR). With the results, meaningful insights are gained by conducting the asymptotic outage analysis. Specifically, it is revealed that full time diversity order can be achieved, while full spatial diversity order is unreachable as compared to MIMO-HARQ systems without keyhole effect. Moreover, we prove that the asymptotic outage probability is a monotonically increasing and convex function of the transmission rate. More importantly, although HARQ-IR performs better than HARQ-CC owing to its higher coding complexity, this advantage becomes negligible in the large-scale array regime. Finally, the numerical results are verified by Monte-Carlo simulations along with some in-depth discussions.
\end{abstract}
\IEEEpeerreviewmaketitle
\hyphenation{HARQ}
\section{Introduction}\label{sec:int}
\IEEEPARstart{W}{ith} the rapid development of intelligent transportation systems, vehicle-to-vehicle (V2V) communications have been widely studied in recent years. {\color{blue}Specifically, V2V communications only allow the exchange of information between adjacent vehicles with short distance,  which enhances the transmission reliability, supports delay-sensitive applications, and improves traffic safety \cite{jie2018latency}.} {\color{blue} However, V2V communications considerably differ from mobile cellular communications. On one hand, both transmit and receive vehicles are in motion, which results in more significant Doppler effects and more rapid channel dynamics than cellular communications. Thus, V2V communications usually undergo time-varying fading channels, which lead to frequent estimations of channel state information (CSI). On the other hand, the transceiver antennas are mounted at almost the same height, where the local scatters in the surrounding environment incur more than one multiplicative small-scale fading processes, i.e., cascaded fading \cite{alghorani2020improved}. The local scattering objects include buildings, vehicles, street corners, tunnels, etc., which obstruct the direct link between two vehicles and lead to non-line-of-sight (NLOS) propagation channel condition \cite{matolak2011worse} (cf. Fig. 1). To characterize cascaded fading in V2V communications, the authors in \cite{ali2018throughput,fu2018applying,sen2008vehicle,bithas2017double} proposed to use double-bounce/multiple scattering distributions, such as double-Rayleigh, double-Nakagami-$m$, double-Weibull, and double-generalized Gamma distributions. Moreover, \color{blue}{a decode-and-forward relaying scheme was developed in \cite{bithas2017v2v} for V2V communications by considering double-Nakagami fading, where both the exact and asymptotic outage probabilities were derived.} It was demonstrated experimentally in \cite{matolak2011worse} that the double-bounce scattering distributions can provide an accurate statistical fit for channel modeling of V2V communications. Both experimental and theoretical results verified that the error performance of the multiple scattering model is worse than that of the traditional Rayleigh channel model for cellular communication systems. Therefore, V2V communications often suffer from a more severe fading leading to lower spectral efficiency and reception reliability than cellular communications.}

To boost the spectral efficiency and reliability, multi-input multi-output (MIMO) aided V2V communications have drawn an ever-increasing attention, because multiple antennas can be easily placed on vehicles with large surface \cite{alghorani2017on,fu2016ber,yang2021asymptotic}. {\color{blue}In contrast to single-input single-output (SISO) systems, MIMO systems equipping with multiple antennas are capable of reaping the benefit of spatial multiplexing gain.} Nevertheless, in realistic propagation environments, the performance of MIMO systems is also susceptible to multiple scattering propagation. {\color{blue}Particularly for MIMO assisted V2V (MIMO-V2V) communications, multiplicative fading processes encountered in multiple scattering condition are inevitable due to the mobility of the vehicles and low elevation height of the transceiver antennas. In dense urban environments, all the MIMO-V2V propagation paths travel through the same narrow pipe, which results in the so-called keyhole effect \cite{alghorani2017on}.} {\color{blue}In practice, the keyhole effect is an important and non-negligible characteristic of MIMO-V2V channels that jeopardizes the diversity gain \cite{talha2011channel,ngo2017no}.
The keyhole effect brings about the cascaded fading for MIMO-V2V channels. In particular, the coefficient matrix of the keyhole MIMO channel is expressed as a product of those of the multiple-input single-output (MISO) channel from transmitter to keyhole and the single-input multiple-output (SIMO) channel from keyhole to receiver. This introduces the spatial correlation among MIMO-V2V channels and rank deficiency issue, which degrade the capacity and outage performance of MIMO communications \cite{almers2006keyhole}.} \textcolor{blue}{To investigate the keyhole effect, the ergodic capacity of MIMO systems and the average symbol error rate (SER) of space-time block codes (STBC) were derived in closed-form in \cite{shin2003capacity,shin2003effect}, respectively}. In \cite{sanayei2007antenna}, the system performance for antenna selection under MIMO-keyhole channel was studied. {\color{blue} In \cite{taniguchi2011statistical}, a statistical analysis of signal-to-noise ratio (SNR) was carried out for MIMO keyhole channels by considering double-Rayleigh and double-Nakagami-Rice fadings.} In \cite{ye2020asymptotic}, the authors derived the deterministic approximation for the ergodic rate of a large scale MISO system over keyhole fading channel, where the maximum ratio transmission (MRT) precoding was adopted. Furthermore, the impact of the keyhole channels on satellite communications has been studied in \cite{goswami2019satellite}. Apart from keyhole fading channels in outdoor environment, the authors in \cite{das2018blind} found that furniture, windows and doors can also incur the keyhole effect. Moreover, it was found in \cite{yang2020accurate} that  the keyhole effect also appears in  dual-hop reconfigurable intelligent surface (RIS) aided wireless systems.  All the relevant works, i.e., \cite{matolak2011worse,chizhik2002keyholes,ngo2017no,shin2003capacity,shin2003effect,levin2008on,ye2020asymptotic,goswami2019satellite,
das2018blind,yang2020accurate,sanayei2007antenna}, have demonstrated that keyhole channels not only offset the advantage of spatial diversity of MIMO, but also degrade the spatial multiplexing gain. {\color{blue} Since the keyhole effect negatively impacts the reliability of MIMO-V2V communications, it is of necessity to remedy the performance loss for MIMO-V2V communications.}

To address the above issue, hybrid automatic repeat request (HARQ) is a promising technique to support reliable communications \cite{ding2017on,shi2018cooperative,cai2018performance,shi2019effective,wang2020outage}. Specifically, the essence of HARQ is utilizing both forward error control and automatic repeat request \cite{dahlman20103g}. {\color{blue}As opposed to the adaptive modulation and coding (AMC) scheme that requires perfect instantaneous CSI at the transmitter (CSIT), HARQ only needs the statistical/outdated CSIT by relying on the acknowledgement feedback. As aforementioned, the feature of rapid channel dynamics emerged in V2V communications entails excessive CSIT acquisition overhead. By comparing to AMC, the adoption of HARQ in V2V communications is favorable for overcoming channel uncertainties. Moreover, HARQ facilitates the implementation of MIMO in V2V communications thanks to its neglected signaling overhead.} Based on whether the erroneously received packets are discarded or not, and what types of coding\&decoding strategies are used, HARQ can be further divided into the following three basic schemes, namely, Type-I HARQ, HARQ with chase combining (HARQ-CC) and HARQ with incremental redundancy (HARQ-IR). In particular, Type-I HARQ performs the decoding based on the currently received packet without storing the failed packets. In contrast to Type-I, both HARQ-CC and HARQ-IR schemes save the failed packets and carry out the joint decoding with the subsequent packets by using maximum ratio combining (MRC) and code combining, respectively. Thanks to the outstanding potential of HARQ schemes, they have been widely adopted to assist MIMO communications. For example, in \cite{shen2011hybrid}, by considering both HARQ-CC and HARQ-IR schemes, the fundamental performance limits and linear dispersion code design for the MIMO systems were studied. Additionally, by aiming at the maximization of  energy efficiency, the HARQ-IR assisted massive MIMO systems were investigated in \cite{kim2018energy}. {\color{blue} Nevertheless, the performance of HARQ schemes over keyhole fading channels was seldom reported in the literature except for \cite{ali2018throughput} and \cite{chelli2014performance}. In \cite{ali2018throughput} and \cite{chelli2014performance}, the outage probability, throughput and delay analyses were conducted for the SISO-HARQ-CC and SISO-HARQ-IR schemes over double-Rayleigh fading channels, respectively. It is worth mentioning that keyhole MIMO fading encompasses double-Rayleigh fading as a special case by assuming only a single antenna at the transceiver. In order to extend the application of MIMO to V2V communications, this paper focuses on the performance investigation of MIMO-HARQ systems over keyhole fading channels.}

Since the outage probability is the key performance matric, this paper thoroughly investigates the outage performance of MIMO-HARQ assisted V2V communications with keyhole effect, where three different HARQ schemes are considered. However, the cascaded property of fading channels and the complexity of HARQ schemes significantly challenge the outage analysis. {\color{blue}To the best of our knowledge, this is the first treatise that touches upon MIMO-HARQ communications with keyhole effect, and many helpful physical insights will be extracted by conducting asymptotic outage analysis in the high SNR and large-scale array regimes.} Specifically, the main contributions of this paper can be summarized as follows.
\begin{enumerate}
  \item Closed-form expressions are derived for the outage probabilities of MIMO-HARQ assisted V2V communications with keyhole effect by using integral transforms, e.g., moment generating function (MGF) and Mellin transform. {\color{blue} Besides, the outage expressions of MIMO-HARQ-CC and MIMO-HARQ-IR schemes in our work collapse to those of SISO-HARQ schemes obtained by \cite{ali2018throughput} and \cite{chelli2014performance}, respectively.}
  \item In order to reveal more insights, the asymptotic outage analyses are conducted in this paper. With asymptotic results, it is concluded that full spatial diversity order is unreachable with MIMO, while full time diversity order can be achieved from using HARQ. This obviously justifies the effectiveness of using HARQ to conquer keyhole effect. More specifically, the spatial diversity order is determined by the minimum of the numbers of transmit and receive antennas. Moreover, it is proved that the asymptotic outage probability is a monotonically increasing and convex function with respect to (w.r.t.) transmission rate. This property facilitates the optimal rate selection for practical system design.
  \item More interestedly, it is found that the MIMO-HARQ-CC assisted V2V communications are able to achieve a comparable performance as the MIMO-HARQ-IR assisted ones in the large-scale array regime. This indicates that HARQ-CC is more effective than HARQ-IR for massive MIMO systems due to its lower computational complexity and hardware requirement.
\end{enumerate}

The remainder of this paper is organized as follows. Section \ref{sec:sys_mod} introduces the system model of MIMO-HARQ assisted V2V communications with keyhole effect. In Section \ref{sec:out}, the outage analysis are conducted to obtain the exact and asymptotic expressions of the outage probabilities. Additionally, with the asymptotic results, some profound discussions are undertaken in Section \ref{sec:div}. In Section \ref{sec:num}, numerical results are presented for verifications. Finally, Section \ref{sec:con} concludes this paper.

\emph{Notation}: The following notations will be used throughout this paper. Bold uppercase and lowercase letters are used to denote matrices and vectors, respectively. ${\bf X}^{\mathrm{H}}$, ${\bf X}^{-1}$, ${\mathrm{det}}({\bf X})$ and ${\rm tr}({\bf X})$ stand for the conjugate transpose, the inverse, the  determinant and the trace of matrix ${\bf X}$, respectively. $\mathbf{I}$ represents an identity matrix. $\left\| \cdot \right\|$ denotes the Euclidean norm of a vector. ${\mathcal {CN}}({\bf{0}},{\bf I})$ represents the  complex Gaussian vector with zero mean vector and identity covariance matrix. $\mathbf{A}\succ\mathbf{B}$ means that $\mathbf{A}-\mathbf{B}$ is a positive definite matrix. ${\rm i}=\sqrt{-1}$ denotes the imaginary unit. The symbol ``$\simeq$'' denotes ``asymptotically equal to''.
The definitions of any other notations are deferred to the place where they arise.
\section{System Model}\label{sec:sys_mod}
\begin{figure}
  \centering
  \includegraphics[width=3.5in]{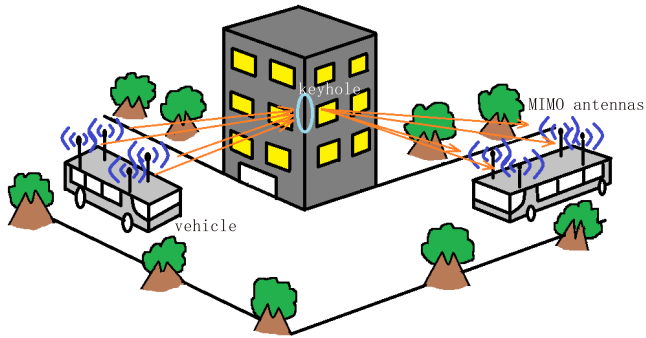}
  \caption{An example for HARQ-MIMO assisted V2V communications with keyhole effect.}\label{fig:model}
\end{figure}
As shown in Fig. \ref{fig:model}, we consider a MIMO-HARQ V2V communication system under urban environments where the transmit vehicle and the receive vehicle are equipped with $N_{T}$ and $N_{R}$ antennas, respectively. We assume that there are a number of obstacles between the transmitter and receiver, and the transmitted signal propagates through electromagnetically small apertures (or keyholes) among obstacles. By following the keyhole channel model in \cite{shin2003effect}, the channel matrix ${\bf H}$ for MIMO-HARQ V2V communications is given by
\begin{equation}\label{eqn:channel_model}
{{\bf{H}}} = {{\bf{u}}}{{\bf{v}}}^{\rm{H}} = \left( {\begin{array}{*{20}{c}}
{{u_{1}}{v_{1}^\ast}}&{{u_{1}}{v_{2}^\ast}}& \cdots &{{u_{1}}{v_{N_T}^\ast}}\\
{{u_{2}}{v_{1}^\ast}}&{{u_{2}}{v_{2}^\ast}}& \cdots &{{u_{2}}{v_{N_T}^\ast}}\\
 \vdots & \vdots & \ddots & \vdots \\
{{u_{{N_R}}}{v_{1}^\ast}}&{{u_{{N_R}}}{v_{2}^\ast}}& \cdots &{{u_{{N_R}}}{v_{N_T}^\ast}}
\end{array}} \right),
\end{equation}
where each entry of ${\bf{v}}$ and ${\bf{u}}$ follows independent and identically distributed (i.i.d.) complex normal distribution, i.e., ${\bf{v}} \sim {\mathcal {CN}}({\bf{0}},{\bf I}_{N_T})$ and ${\bf{u}} \sim {\mathcal {CN}}({\bf{0}},{\bf I}_{N_R})$.

Clearly from \eqref{eqn:channel_model}, ${{\bf{H}}}$ is a rank-one matrix in the presence of the keyhole effect. The rank-one deficiency significantly reduces the spatial multiplexing gain and degrades the MIMO capacity. To enhance the reception reliability, three types of HARQ schemes, i.e., Type-I HARQ, HARQ-CC and HARQ-IR, are employed in this paper. Therefore, the accumulated mutual information obtained by the three types of HARQ schemes after $K$ HARQ rounds can be expressed as \cite{7438859}
\begin{equation}\label{eqn:mu}
I = \left\{ {\begin{array}{*{20}{c}}
{\max\limits_{k=1,\cdots,K}{\log _2}\left( {\det \left[ {\mathbf{I}_{{N_R}}} + \frac{{{\gamma _k}}}{{{N_T}}}{{\bf{H}}_k}{{\bf{H}}_k}^{\rm{H}} \right]} \right),}&{\textup{Type-I}},\\
{{\log _2}\left( {\det \left[ {{{\bf I}_{{KN_R}}} + \mathbf{H}_{CC}{\mathbf{H}_{CC}}^{\rm{H}}} \right]} \right),}&{\textup{HARQ-CC}},\\
{\sum\limits_{k = 1}^K {{{\log }_2}\left( {\det \left[ {{{\bf I}_{{N_R}}} + \frac{{{\gamma _k}}}{{{N_T}}}{{\bf{H}}_k}{{\bf{H}}_k}^{\rm{H}}} \right]} \right)}},&{\textup{HARQ-IR}},
\end{array}} \right.
\end{equation}
where ${\bf{H}}_k,~k\in\{1,\cdots,K\}$ represents the keyhole channel matrix for the $k$-th HARQ transmission round and is modelled according to (\ref{eqn:mu}), $\gamma _k$ stands for the average transmit SNR at the transmitter for the $k$-th round, and $\mathbf{H}_{CC}=[{\sqrt {{{{\gamma _1}}}/{{{N_T}}}}{{\bf{H}}_1}^{\rm{H}}}, \cdots,{\sqrt {{{{\gamma _K}}}/{{{N_T}}}}{{\bf{H}}_K}^{\rm{H}}}]^{{\rm{H}}}$. Besides, ${\bf{H}}_1, \cdots, {\bf{H}}_K$ are assumed to be independent random matrices.

\textcolor[rgb]{0.00,0.07,1.00}{Clearly from (\ref{eqn:channel_model}), since each entry of the channel matrix is expressed as a product of two complex Gaussian random variables, this complex form hinders the following outage analyses. Besides, the accumulated mutual information for the three types of HARQ schemes involves many matrix operations,
such as determinant operations, product of block matrix, which further complicates the analysis. To the best of our knowledge, there is no available results for the outage performance of MIMO-HARQ V2V communications in the literature.}
\section{Analysis of Outage Probability}\label{sec:out}
\textcolor[rgb]{0.00,0.07,1.00}{To investigate the performance of MIMO-HARQ assisted V2V communication systems by considering the three different HARQ schemes, the outage probability is the most significant and essential performance metric}. More specifically, the outage probability is defined as the probability of the event that the accumulated mutual information is less than the preset transmission rate $R$. Accordingly, the outage probability of MIMO-HARQ assisted V2V communication systems can be written as
\begin{equation}\label{eqn:op}
P_{out} = \Pr \left( {I < R} \right).
\end{equation}
By substituting (\ref{eqn:mu}) into (\ref{eqn:op}), the outage analyses for the three types of MIMO-HARQ schemes will be undertaken individually in the following. \textcolor[rgb]{0.00,0.07,1.00}{Furthermore, other performance metrics such as throughput, delay, and ergodic capacity can be expressed in terms of the outage probability. For instance, the throughput of HARQ can be obtained by using \cite[Eq. (5)]{chelli2014performance}. In addition, the ergodic capacity is defined as the average mutual information \cite{muller2007characterization}. The ergodic capacity can be expressed in terms of the MGF of the probability density function (PDF) of the mutual information $I$, which is determined by the distribution of $I$ in \eqref{eqn:op} \cite{taricco2018outage}. We omit the detailed discussions due to the page limitation.}
\subsection{MIMO-Type-I HARQ}
For Type-I HARQ scheme, the erroneously received packets are discarded and negative acknowledgement is sent back to request the retransmission of the message until the maximum number of transmissions, i.e., $K$. By substituting the  accumulated mutual information of Type I-HARQ scheme (\ref{eqn:mu}) into (\ref{eqn:op}) and combining with MIMO keyhole channel model (\ref{eqn:channel_model}), the outage probability of MIMO-Type-I HARQ scheme can be rewritten as
 \begin{align}\label{eqn:out_probI}
P_{out}^{Type-I}&=\Pr \left( {\max\limits_{k=1,...,K}{\log _2}\left( {\det \left[ {\mathbf{I}_{{N_R}}} + \frac{{{\gamma _k}}}{{{N_T}}}{{\bf{H}}_k}{{\bf{H}}_k}^{\rm{H}} \right]} \right) < R} \right) \notag\\
&= \prod\limits_{k = 1}^K {\Pr \left( {{{\log }_2}\left( {1 + \frac{{{\gamma _k}}}{{{N_T}}}{{\left\| {{{\bf{u}}_k}} \right\|}^2}{{\left\| {{{\bf{v}}_k}} \right\|}^2}} \right) < R} \right)} \notag\\
&=\prod\limits_{k = 1}^K {{F_{{X_k}}}\left( {\frac{{{N_T}}}{{{\gamma _k}}}\left( {{2^R} - 1} \right)} \right)},
\end{align}
where the second step holds by using $\mathrm{det}(\mathbf{I}+\mathbf{A}\mathbf{B})=\mathrm{det}(\mathbf{I}+\mathbf{B}\mathbf{A})$, and  ${F_{{X_k}}}(x)$ denotes the cumulative distribution function
(CDF) of $X_k = {{\left\| {{{\bf{u}}_k}} \right\|}^2}{{\left\| {{{\bf{v}}_k}} \right\|}^2}$. From (\ref{eqn:out_probI}), the outage analysis of MIMO-Type-I HARQ scheme boils down to determining the distribution of the equivalent channel gain ${{\left\| {{{\bf{u}}_k}} \right\|}^2}{{\left\| {{{\bf{v}}_k}} \right\|}^2}$. Since ${\left\| {{{\bf{u}}_k}} \right\|}^2$ and ${\left\| {{{\bf{v}}_k}} \right\|}^2$ are central chi-square random variables with $2N_R$ and $2N_T$ degrees of freedom, respectively. The PDF of $X_k$ is given as \cite{shin2003capacity}
\begin{align}\label{eqn:pdf_X_l}
{f_{{X_k}}}\left( x \right)=  \frac{{2{x^{\left( {{N_T} + {N_R}} \right)/2 - 1}}{K_\tau}\left( {2\sqrt x } \right)}}{{\Gamma \left( {{N_T}} \right)\Gamma \left( {{N_R}} \right)}},
\end{align}
 where $\Gamma(\cdot)$ is the gamma function \cite[Eq. (8.310)]{jeffrey2007table}, $K_{\tau}(\cdot)$ represents the modified
Bessel function of the $\tau$-th order \cite[Eq. (8.432.6)]{jeffrey2007table}, and $\tau=\left| {{N_T} - {N_R}} \right|$. \textcolor[rgb]{0.00,0.07,1.00}{In analogy to \cite{adamchik1990algorithm}, it is suggested to invoke Meijer G-function to generalize our analytical results}. By utilizing \cite[Eq. (9.34.3)]{jeffrey2007table}, the PDF of $X_k$ can be expressed in the form of Meijer G-function as \cite[Eq. (9.301)]{jeffrey2007table}
\begin{align}\label{eqn:pdf_X_l_Meij}
{f_{{X_k}}}\left( x \right) &= \frac{{G_{0,2}^{2,0}\left( {\left. {\begin{array}{*{20}{c}}
 - \\
{{N_T},{N_R}}
\end{array}} \right|x} \right)}}{{x\Gamma \left( {{N_T}} \right)\Gamma \left( {{N_R}} \right)}}.
\end{align}
Then, the CDF of $X_k$ can be derived by using \cite[Eq. (26)]{adamchik1990algorithm} as
\begin{equation}\label{eqn:CDF_Zl}
{F_{{X_k}}}\left( x \right) =\int_0^x {{f_{{X_k}}}\left( y \right)dy}= \frac{{G_{1,3}^{2,1}\left( {\left. {\begin{array}{*{20}{c}}
1\\
{{N_T},{N_R},0}
\end{array}} \right|x} \right)}}{{\Gamma \left( {{N_T}} \right)\Gamma \left( {{N_R}} \right)}}.
\end{equation}
Finally, substituting  (\ref{eqn:CDF_Zl}) into (\ref{eqn:out_probI}) yields the closed-form expression of outage probability for the MIMO-Type-I HARQ scheme as
\begin{align}\label{eqn:op_type_I}
P_{out}^{Type-I}
=&\frac{1}{\left({\Gamma \left( {{N_T}} \right)\Gamma \left( {{N_R}} \right)}\right)^K}\prod\limits_{k = 1}^K {G_{1,3}^{2,1}\left( {\left. {\begin{array}{*{20}{c}}
1\\
{{N_T},{N_R},0}
\end{array}} \right|{\frac{{{N_T}}}{{{\gamma _k}}}\left( {{2^R} - 1} \right)}} \right)}.
\end{align}
However, although the Meijer G-function in (\ref{eqn:op_type_I}) is a built-in function in many popular mathematical software packages, such as MATLAB, the integral form of Meijer G-function are still complex which hampers the extraction of useful insights, such as diversity order and coding gain. To overcome this issue, an asymptotic expression of the asymptotic outage probability is obtained in the high SNR regime, as shown in the following theorem.
\begin{theorem}\label{The:asy_I}
In the high SNR regime, i.e., $\gamma_k\to\infty$, the asymptotic outage probability of MIMO-Type-I HARQ scheme is given by
\begin{equation}\label{eqn:asy I}
P_{out}^{Type-I} \simeq \left\{ {\begin{array}{*{20}{c}}
{\prod\limits_{k = 1}^K {\frac{{{{\left( {\frac{{{N_T}}}{{{\gamma _k}}}\left( {{2^R} - 1} \right)} \right)}^{{N_T}}}\ln {{\gamma _k}} }}{{{N_T}({{\Gamma}\left( {{N_T}} \right)}) ^{2}}}},}&{\textup{$N_T - N_R = 0$}},\\
{\prod\limits_{k = 1}^K {\frac{{\Gamma \left( {\tau} \right)}}{{\Gamma \left( {{N_T}} \right)\Gamma \left( {{N_R}} \right)}}\frac{{{{\left( {\frac{{{N_T}}}{{{\gamma _k}}}\left( {{2^R} - 1} \right)} \right)}^{\left( {{N_T} + {N_R}} \right)/2 -\tau/2}}}}{{\left( {{N_T} + {N_R}} \right)/2 - \tau/2}}},}&{\textup{$N_T - N_R \ne 0$}}.\\
\end{array}} \right.
\end{equation}
\end{theorem}
\begin{proof}
Please see Appendix \ref{eqn:8}.
\end{proof}
\subsection{MIMO-HARQ-CC}\label{HARQ-CC}
With regard to the MIMO-HARQ-CC scheme,  the same packet is retransmitted, all the previously failed packets are stored for the subsequent decoding. Based on (\ref{eqn:mu}) together with the definition of $\mathbf{H}_{CC}$, the accumulated mutual information for the MIMO-HARQ-CC scheme can be rewritten as
\begin{align}\label{eqn:ICC_fur}
I_{CC}
 &= {\log _2}\left( {\det \left( {{{\bf{I}}_{N_T}} + \sum\limits_{k = 1}^K {\frac{{{\gamma _k}}}{{{N_T}}}{\left\| {{{\bf{u}}_k}} \right\|^2}{{\bf{v}}_k}{{\bf{v}}_k}^{\rm{H}}} } \right)} \right).
\end{align}
Clearly, since ${\left\| {{{\bf{u}}_k}} \right\|^2}$ is a central chi-square random variable and ${{\bf{v}}_k}{{{\bf{v}}_k}^{\rm{H}}}$ is a complex central Wishart matrix with one degree of freedom, the complex form of \eqref{eqn:ICC_fur} makes the derivation of the distribution of $I_{CC}$ fairly intractable. {\color{blue} Unfortunately, the integral transform based approaches developed in \cite{yang2021asymptotic,shin2003effect,shi2020achievable,shi2017asymptotic} are inapplicable herein.} Hence, we resort to a lower bound to approximate $I_{CC}$, as shown in the following Lemma.
\begin{lemma}\label{lemma1}
The accumulated mutual information of the MIMO-HARQ-CC scheme is lower bounded as
\begin{align}\label{eqn:A_K_det_lower}
&I_{CC}\geq  {\log _2}\left(1 + \sum\limits_{k = 1}^K {\frac{{{\gamma _k}}}{{{N_T}}}{\left\| {{{\bf{u}}_k}} \right\|^2}{{\left\| {{{\bf{v}}_k}} \right\|}^2}}\right).
\end{align}
\end{lemma}
\begin{proof}
By defining ${{\bf{A}}_{n}} = {{\bf{I}}_{N_T}} + \sum\nolimits_{k = 1}^{n} {\frac{{{\gamma _k}}}{{{N_T}}}{\left\| {{{\bf{u}}_k}} \right\|^2}{{\bf{v}}_k}{{\bf{v}}_k}^{\rm{H}}}$, $n\in\{0,\cdots,K-1\}$, it follows that
\begin{align}\label{eqn:proof}
I_{CC}&={\log _2}\left(\det \left( {{\bf{A}}_{K - 1}} \left(1+\frac{{{\gamma _K}}}{{{N_T}}}{\left\| {{{\bf{u}}_K}} \right\|^2}{{\bf{v}}_K}^{\rm{H}}{\bf{A}}_{K - 1}^{ - 1}{{\bf{v}}_K}\right)\right)\right)\notag\\
&={\log _2}\left( {\det \left( {{{\bf{A}}_{K - 1}}} \right)\left( {1 + \frac{{{\gamma _K}}}{{{N_T}}}{{\left\| {{{\bf{u}}_K}} \right\|}^2}{{\bf{v}}_K}^{\rm{H}}{\bf{A}}_{K - 1}^{ - 1}{{\bf{v}}_K}} \right)} \right)\notag\\
&\geq {\log _2}\left(\det \left( {{{\bf{A}}_{K - 1}}} \right) + \frac{{{\gamma _K}}}{{{N_T}}}{\left\| {{{\bf{u}}_K}} \right\|^2}{{\left\| {{{\bf{v}}_K}} \right\|}^2}\right)\geq \cdots\notag\\
&\geq   {\log _2}\left(\det \left( {{{\bf{A}}_{n}}} \right) + \sum\limits_{k = n+1}^K {\frac{{{\gamma _k}}}{{{N_T}}}{\left\| {{{\bf{u}}_k}} \right\|^2}{{\left\| {{{\bf{v}}_k}} \right\|}^2}}\right),
\end{align}
where the inequality holds because ${{\bf{A}}_{n}}$ is a positive definite matrix with all eigenvalues no less than $1$. Thus we have $ {\bf{A}}_{n}^{ - 1} \succ {{\bf{I}}_{N_T}}/{{\det \left( {{{\bf{A}}_{n}}} \right)}}$. By repeatedly applying this property, \eqref{eqn:A_K_det_lower} can be finally obtained.
\end{proof}

By substituting (\ref{eqn:A_K_det_lower}) into (\ref{eqn:op}), the upper bound of outage probability for the MIMO-HARQ-CC scheme is given by
 \begin{align}\label{eqn:out_CC}
P_{out}^{CC}
&\leq\Pr \left( {\log _2}\left(1 + \sum\limits_{k = 1}^K {\frac{{{\gamma _k}}}{{{N_T}}}{\left\| {{{\bf{u}}_k}} \right\|^2}{{\left\| {{{\bf{v}}_k}} \right\|}^2}}\right) < R \right) \notag\\
& =  {F_{{X_{CC}}}}\left( {2^R} - 1 \right),
\end{align}
where ${X_{CC}}= \sum\nolimits_{k = 1}^K {{{\gamma _k}/{N_T}}{\left\| {{{\bf{u}}_k}} \right\|^2}{{\left\| {{{\bf{v}}_k}} \right\|}^2}}= \sum\nolimits_{k = 1}^K{{\gamma _k}/{N_T}} X_k$. To determine the CDF of ${X_{CC}}$, it is essential to obtain the distribution of the sum of multiple random variables $X_1, \cdots, X_K$, which can be tackled by using the method of MGF \cite{shi2020achievable}. As a result, the CDF of ${X_{CC}}$ is given by the following theorem.
\begin{theorem}\label{The:pcdf_Y_K}
The  CDF of ${X_{CC}}$  for the cases of $N_{T} \ge N_{R}$ and $N_{T} < N_{R}$ are given by
\begin{equation}\label{eqn:cdf_Yk}
{F_{{X_{CC}}}}\left( x \right) = \left\{ {\begin{array}{*{20}{c}}
{\prod\limits_{k = 1}^K {{{\left( {\frac{{{N_T}}}{{{\gamma _k}}}} \right)}^{{N_T}}}} \frac{1}{{2\pi {\rm i}}}\int\limits_{a - {\rm i}\infty }^{a + {\rm i}\infty } {\frac{{\prod\limits_{k = 1}^K {\Psi \left( {{N_T},\tau+1;\frac{{{N_T}}}{{{\gamma _k}}}{s^{ - 1}}} \right)} }}{{{s^{K{N_T} + 1}}}}{e^{sx}}ds},}&{\textup{$N_{T} \ge N_{R}$}}\\
 {\prod\limits_{k = 1}^K {{{\left( {\frac{{{N_T}}}{{{\gamma _k}}}} \right)}^{{N_R}}}} \frac{1}{{2\pi {\rm i}}}\int\limits_{a - {\rm i}\infty }^{a + {\rm i}\infty } {\frac{{\prod\limits_{k = 1}^K {\Psi \left( {{N_R},\tau+1;\frac{{{N_T}}}{{{\gamma _k}}}{s^{ - 1}}} \right)} }}{{{s^{K{N_R} + 1}}}}{e^{sx}}ds},}&{\textup{$N_{T} < N_{R}$}},
\end{array}} \right.
\end{equation}
where $\Psi \left( \alpha,\gamma;z \right)$ denotes the Tricomi confluent hypergeometric function \cite[eq. (9.211.4)]{jeffrey2007table}. It is worth mentioning that (\ref{eqn:cdf_Yk}) can be evaluated via numerical inversion of Laplace transform \cite{abate1995numerical}.
\end{theorem}
\begin{proof}
Please see Appendix \ref{eqn:CC}.
\end{proof}
Accordingly, substituting $x=2^R-1$ into (\ref{eqn:cdf_Yk}) gives rise to the upper bound of the outage probability for the MIMO-HARQ-CC scheme.

Similarly to Theorem \ref{The:asy_I}, the asymptotic outage analysis in the high SNR regime is performed for the MIMO-HARQ-CC scheme, as given by the following theorem.
\begin{theorem}\label{lem:op-cc}
Under high SNR, i.e., $\gamma_k\to\infty$, the CDF of $X_{CC}$ asymptotically behaves as
\begin{equation}\label{eqn:as_cc}
{F_{{X_{CC}}}}\left( x \right) \simeq \left\{ {\begin{array}{*{20}{c}}
{\frac{1}{{\left( {K{N_R}} \right)!}}\prod\limits_{k = 1}^K {\frac{{\Gamma \left(\tau \right)}}{{\Gamma \left( {{N_T}} \right)}}{{\left( {\frac{{{N_T}}}{{{\gamma _k}}}x} \right)}^{{N_R}}}},}&{\textup{$N_{T} > N_{R}$}},\\
{\frac{1}{{\left( {K{N_T}} \right)!}}\prod\limits_{k = 1}^K {\frac{{\Gamma \left( \tau \right)}}{{\Gamma \left( {{N_R}} \right)}}{{\left( {\frac{{{N_T}}}{{{\gamma _k}}}x} \right)}^{{N_T}}}},}&{\textup{$N_{T} < N_{R}$}},\\
{\frac{1}{{\left( {K{N_T}} \right)!}}\prod\limits_{k = 1}^K {\frac{{\ln  \gamma_k }}{{\Gamma \left( {{N_T}} \right)}}{{\left( {\frac{{{N_T}}}{{{\gamma _k}}}x} \right)}^{{N_T}}}},}&{\textup{$N_{T} = N_{R}$}},\\
\end{array}} \right.
\end{equation}
and the asymptotic outage probability for the MIMO-HARQ-CC scheme can be directly obtained by replacing $x$ in (\ref{eqn:as_cc}) with $2^R-1$.
\end{theorem}
\begin{proof}
Please see Appendix \ref{app:op-CC}.
\end{proof}
\subsection{MIMO-HARQ-IR}
Different from HARQ-CC, HARQ-IR employs code combining for joint decoding with the erroneously received packets. Besides, the redundant information is incrementally transmitted in each HARQ round. By substituting the accumulated mutual information for MIMO-HARQ-IR into (\ref{eqn:op}), the corresponding outage probability is rewritten as
\begin{align}\label{eqn:out_prob}
P_{out}^{IR}
&= \Pr \left( {\sum\limits_{k = 1}^K {{{\log }_2}\left( {\det \left( {{{\bf I}_{{N_R}}} + \frac{{{\gamma _k}}}{{{N_T}}}{{\bf{H}}_k}{{\bf{H}}_k}^{\rm{H}}} \right)} \right)}< R} \right) \notag\\
&= \Pr \left( \prod\limits_{k = 1}^K {\left( {1 + \frac{{{\gamma _k}}}{{{N_T}}}{X_k}} \right)}  < {2^R}\right) \buildrel \Delta \over = {F_{{X_{IR}}}}\left( {{2^R}} \right),
\end{align}
where $F_{X_{IR}}(x)$ represents the CDF of $X_{IR}=\prod\nolimits_{k = 1}^K {\left( {1 + {{{\gamma _k}}}/{{{N_T}}}{X_k}} \right)}$.

It is easily found that $X_{IR}$ is the product of multiple independent random variables. Hence, we recourse to Mellin transform which is widely adopted to derive the distribution of the product of multiple independent random variables \cite{yang2021asymptotic,shi2017asymptotic}. Specifically, the Mellin transform is a kind of integral transform like the Laplace and Fourier transforms. The Mellin transform has the property that the Mellin transform of the PDF of the product of multiple independent random variables is equal to the product of their Mellin transforms. With this property, the CDF of $X_{IR}$ is obtained in the following theorem.
\begin{theorem}\label{lem:pcdf_Z_K}
The CDF of $X_{IR}$ is written in terms of an inverse Mellin transform as
\begin{align}\label{eqn:CDF_Z_k_der}
{F_{X_{IR}}}\left( x \right)
& = \frac{1}{{2\pi {\rm i}}}\int\limits_{ - c - {\rm i}\infty }^{ - c + {\rm i}\infty } {\frac{{{x^s}}}{s}\prod\limits_{k = 1}^K {\frac{{G_{1,3}^{3,1}\left( {\left. {\begin{array}{*{20}{c}}
1\\
{s,{N_T},{N_R}}
\end{array}} \right|\frac{{{N_T}}}{{{\gamma _k}}}} \right)}}{{\Gamma \left( {{N_T}} \right)\Gamma \left( {{N_R}} \right)\Gamma \left( {s} \right)}}} ds}.
\end{align}
\end{theorem}
\begin{proof}
Please see Appendix \ref{eqn:IR}.
\end{proof}
Therefore, the outage probability of the MIMO-HARQ-IR scheme can be obtained by substituting $x=2^R$ into (\ref{eqn:CDF_Z_k_der}). Therefore, the outage probability of the MIMO-HARQ-IR scheme can also be expressed as an inverse Laplace transform, which can be numerically evaluated.

However, due to the involvement of Meijer G-function, it is intractable to extract meaningful insights from (\ref{eqn:CDF_Z_k_der}). To address this issue, the asymptotic outage probability of the MIMO-HARQ-IR scheme in the high SNR region is derived in the following theorem.
\begin{theorem}\label{The:IR_2}
Under high SNR, i.e., $\gamma_k\to\infty$, the outage probability of the MIMO-HARQ-IR scheme is asymptotically equal to
\begin{equation}\label{eqn:as_IR}
P_{out}^{IR} \simeq \left\{ {\begin{array}{*{20}{c}}
{\prod\limits_{k = 1}^K {\frac{{\Gamma \left( \tau\right){{\left( {\frac{{{N_T}}}{{{\gamma _k}}}} \right)}^{{N_R}}}}}{{\Gamma \left( {{N_T}} \right)}}} G_{K + 1,K + 1}^{0,K + 1}\left( {\left. {\begin{array}{*{20}{c}}
{1,{N_R} + 1, \cdots ,{N_R} + 1}\\
{ 1, \cdots ,1,0}
\end{array}} \right|2^R} \right),}&{\textup{$N_{T} > N_{R}$}},\\
{\prod\limits_{k = 1}^K {\frac{{\Gamma \left( \tau \right){{\left( {\frac{{{N_T}}}{{{\gamma _k}}}} \right)}^{{N_T}}}}}{{\Gamma \left( {{N_R}} \right)}}} G_{K + 1,K + 1}^{0,K + 1}\left( {\left. {\begin{array}{*{20}{c}}
{1,{N_T} + 1, \cdots ,{N_T} + 1}\\
{ 1, \cdots ,1,0}
\end{array}} \right|2^R} \right),}&{\textup{$N_{T} < N_{R}$}},\\
{\prod\limits_{k = 1}^K {\frac{{{{\left( {\frac{{{N_T}}}{{{\gamma _k}}}} \right)}^{{N_R}}}\ln  {{\gamma _k}} }}{{\Gamma \left( {{N_T}} \right)}}} G_{K + 1,K + 1}^{0,K + 1}\left( {\left. {\begin{array}{*{20}{c}}
{1,{N_R} + 1, \cdots ,{N_R} + 1}\\
{ 1, \cdots , 1,0}
\end{array}} \right|2^R} \right)},&{\textup{$N_{T} = N_{R}$}}.                           \\
\end{array}} \right.
\end{equation}
\end{theorem}
\begin{proof}
Please see Appendix \ref{app:IR_2}.
\end{proof}
Comparing with the exact outage expressions for the three MIMO-HARQ schemes, the corresponding asymptotic expressions, i.e., (\ref{eqn:asy I}), (\ref{eqn:as_cc}), and (\ref{eqn:as_IR}), are more concise and useful, which will be detailed in the next section.
\section{Discussions of Asymptotic Results}\label{sec:div}
To facilitate our discussion, we assume equal transmit SNRs, i.e., $\gamma_1=\cdots=\gamma_K=\gamma$. By identifying the asymptotic results (\ref{eqn:asy I}), (\ref{eqn:as_cc}), and (\ref{eqn:as_IR}) with \cite[eq.(3.158)]{tse2005fundamentals}, \cite[eq.(1)]{wang2003simple}, the asymptotic outage probabilities of MIMO-HARQ systems with keyhole effect as $\gamma\to \infty$ can be generalized as \cite{shi2017asymptotic}
\begin{equation}\label{eqn:pro}
P_{out}=\left(\mathcal{C}(R)\gamma\right)^{-d}(\ln \gamma)^{K\mathbb I(N_t-N_r)}+{o}(\gamma^{-d}(\ln \gamma)^{K\mathbb I(N_t-N_r)}),
\end{equation}
where $\mathbb I(N_t-N_r)$ denotes the indicator function, $\mathcal{C}(R)$ represents the modulation and coding gain, $d$ stands for the diversity order, and $o(\gamma^{-d})$ denotes the higher order terms. In what follows, the diversity order $d$, and the modulation and coding gain $\mathcal{C}(R)$ of the three types of MIMO-HARQ schemes are individually discussed. Moreover, the asymptotic property of MIMO-HARQ schemes in the large-scale array regime is also examined.
\subsection{Diversity Order}
The diversity order measures the degrees of freedom of communication systems, which is defined as the declining slope of the outage probability with regard to the transmit SNR on a log-log scale in the high SNR regime as \cite{chelli2014performance}
\begin{equation}\label{eqn:div}
d =-\lim\limits_{\gamma\to\infty}\frac{{\log}\left(P_{out}\right)}{{\log}\left(\gamma\right)}.
\end{equation}
By substituting the asymptotic outage expressions (i.e., (\ref{eqn:asy I}), (\ref{eqn:as_cc}), and (\ref{eqn:as_IR})) into (\ref{eqn:div}), we can obtain the diversity orders for three MIMO-HARQ schemes are equal and are given by $d=K\mathrm{min}\left(N_T,N_R\right)$, where $K$ and $\mathrm{min}(N_T,N_R)$ represent the achievable time and spatial diversity order, respectively. As proved in \cite{yang2021asymptotic}, the diversity order of MIMO systems without keyhole effect is $N_TN_R$. However, the spatial diversity order of the MIMO-HARQ with keyhole effect is the minimum of the numbers of transmit and receive antennas. Thus, full spatial diversity order is unreachable due to the rank deficiency of the keyhole effect. This result is also consistent with \cite{sanayei2007antenna}. Besides, as substantiated in \cite{shi2017asymptotic}, the time diversity order that can be achieved by HARQ is determined by the maximum number of transmissions, i.e., $K$. Clearly, full time diversity order can be achieved from using HARQ for the proposed scheme.  This confirmed the validity of using HARQ to combat keyhole effect. 
\subsection{Modulation and Coding Gain}\label{Coding_gain}


\textcolor[rgb]{0.00,0.07,1.00}{The modulation and coding gain $\mathcal{C}(R)$ quantifies how much transmit power can be reduced by using a certain modulation and coding scheme (MCS) to achieve the same outage performance. In other words, $\mathcal C(R)$ characterizes how much gain can be benefited from the adopted MCS \footnote{\textcolor[rgb]{0.00,0.07,1.00}{For example, if there are two candidate MCSs for a MIMO-HARQ assisted V2V communication system, each having the modulation and coding gain $\mathcal{C}_i(R)$, where $i$ is the index of the MCS and $i=1,2$. To achieve the same outage target $P_{out} = \varepsilon$, the average transmit SNR for each MCS is required to be $\varepsilon = \left(\mathcal{C}_1(R)\gamma_i\right)^{-d}$, where the outage probability is approximated by using $P_{out} \approx \left(\mathcal{C}(R)\gamma\right)^{-d}$ according to \eqref{eqn:pro} and the term $\ln \gamma$ is ignored due to the dominance of $\gamma$ by comparing to $\ln \gamma$ in the high SNR regime. Thus, if we change the adopted MCS from 1 to 2 while guaranteeing the same outage target, the SNR reduction for MCS 2 is given by $10{\log _{10}}{\gamma _2} - 10{\log _{10}}{\gamma _1} \approx 10{\log _{10}}{{\cal C}_1}(R) - 10{\log _{10}}{{\cal C}_2}(R)$. This result is consistent with the observation in Figs. \ref{fig:R23} and \ref{fig:NT}.}}.} With the asymptotic results, the explicit expression of $\mathcal{C}(R)$ can be obtained for different relationships between $N_T$ and $N_R$. 
To simplify our discussion, we only consider the case of $N_T=N_R$. By plugging (\ref{eqn:asy I}), (\ref{eqn:as_cc}), and (\ref{eqn:as_IR}) into (\ref{eqn:pro}), we can obtain the modulation and coding gain for each MIMO-HARQ scheme in the case of $N_T=N_R$ as
\begin{equation}\label{eqn:coding_gain}
\mathcal{C}(R) = \left\{ {\begin{array}{*{20}{c}}
\frac{1}{{N_{T}}\left(2^R-1\right)}(N_T{({\Gamma \left( {{N_T}} \right)})}^2)^{\frac{1}{N_T}},&{\textup{Type-I}},\\
\frac{1}{{{N_T}({2^R} - 1)}}{\left( {\left( {K{N_T}} \right)!} \right)^{\frac{1}{{K{N_T}}}}}{\left( {\Gamma \left( {{N_T}} \right)} \right)^{\frac{1}{{{N_T}}}}},&{\textup{HARQ-CC}},\\
\frac{1}{{N_T}}g{\left( R \right)^{ - \frac{1}{KN_T}}}{\left( {{{\Gamma \left( {{N_T}} \right)}}} \right)^{\frac{1}{{N_T}}}},&{\textup{HARQ-IR}}, \\
\end{array}} \right.
\end{equation}
where ${g}(R)=G_{K + 1,K + 1}^{0,K + 1}\left( {\left. {\begin{array}{*{20}{c}}
{1,{N_R} + 1, \cdots ,{N_R} + 1}\\
{ 1, \cdots , 1,0}
\end{array}} \right|2^R} \right)$.
The convexity and increasing monotonicity of ${g}(R)$ w.r.t. $R$ has been proved in \cite[Lemma 4]{shi2017asymptotic}. It is evident from (\ref{eqn:coding_gain}) that the modulation and coding gain is a decreasing function w.r.t. the transmission rate $R$. Moreover, the coding gains of the three MIMO-HARQ schemes follow the relationship as $\mathcal{C}_{IR}(R)\ge \mathcal{C}_{CC}(R)\ge \mathcal{C}_{I}(R)$, where $\mathcal{C}_{I}(R)$, $ \mathcal{C}_{CC}(R)$ and $\mathcal{C}_{IR}(R)$ denote the coding gains of Type-I HARQ, HARQ-CC and HARQ-IR, respectively. The relationship between $\mathcal{C}_{CC}(R)$ and $ \mathcal{C}_{I}(R)$ directly follows from \eqref{eqn:coding_gain}. Besides, $\mathcal{C}_{IR}(R)$ is the maximum among them, which can be proved by using the integral form of ${g}(R)$ in \cite[eq. (26)]{shi2017asymptotic}. The similar results also apply to the cases of $N_T>N_R$ and $N_T<N_R$. The details are omitted here to avoid redundancy. \textcolor[rgb]{0.00,0.07,1.00}{In summary, the performance of the MIMO-HARQ schemes commonly depends on their implementation complexity. More specifically, the MIMO-HARQ-IR scheme performs the best in terms of the modulation and coding gain at the cost of its highest coding complexity, while the MIMO-Type-I HARQ scheme exhibits the worst performance due to its simple coding mechanism. Nevertheless, the MIMO-Type-I HARQ scheme has the lowest hardware requirement by comparing to the other two schemes, because it does not require extra buffer to store the failed packets. Furthermore, it is worth highlighting that the MIMO-HARQ-CC scheme can attain a balanced tradeoff between the complexity and the performance.}
\subsection{Large-Scale Array Regime}\label{Large}

It is well known that the HARQ-IR scheme outperforms the HARQ-CC scheme in terms of the spectral efficiency, albeit at the cost of very high coding complexity. It is worth mentioning that the HARQ-CC scheme can achieve a balanced tradeoff between the spectral efficiency and the coding complexity. Fortunately, for large-scale antenna systems, it can be proved that the MIMO-HARQ-CC scheme can achieve a comparable performance as the MIMO-HARQ-IR scheme, as shown in the following theorem. 
\begin{theorem}\label{the:las}
For a large number of antenna arrays at the transmitter, i.e., $N_T\to \infty$, the gap between the accumulated mutual informations of the MIMO-HARQ-CC and the MIMO-HARQ-IR schemes converges almost surely to zero, namely,
\begin{equation}\label{eqn:cov}
\Pr \left\{ {\mathop {\lim }\limits_{{N_T} \to \infty } \left( {{I_{IR}} - {I_{CC}}} \right) = 0} \right\} = 1.
\end{equation}
This result indicates that the MIMO-HARQ-CC scheme can provide a comparable capacity/outage performance as the MIMO-HARQ-IR scheme in the large-scale array regime.
\end{theorem}
\begin{proof}
 Please see Appendix \ref{app:Le_2}.
\end{proof}
With Theorem \ref{the:las}, the MIMO-HARQ-CC scheme is more appealing for practical V2V communications due to its lower computational complexity and hardware requirement by comparing to the MIMO-HARQ-IR scheme.


\section{Numerical Analysis}\label{sec:num}
In this section, numerical results and Monte-Carlo simulations are presented for verifications and discussions. Unless otherwise specified, the system parameters are set as $R=3$ bps/Hz, $N_T=N_R=2$ and $K=3$. Moreover, we assume that the transmit SNR are identical across all HARQ rounds, i.e., $\gamma_1=\cdots=\gamma_K=\gamma$ in the sequel, and the labels ``Sim.'', ``Exa.'' and ``Asy.'' represent the simulated, the exact and the asymptotic outage probabilities, respectively.

According to Section \ref{sec:out}, Figs. \ref{fig:22} - \ref{fig:32} show the simulated, exact and asymptotic outage probabilities of the three MIMO-HARQ schemes versus the transmit SNR under different relationships between the numbers of transmit and receive antennas. It is observed from Figs. \ref{fig:22} - \ref{fig:32} that the exact results match well with the simulated ones for the MIMO-Type-I and the MIMO-HARQ-IR schemes. Moreover, it can be seen from Figs. \ref{fig:22} - \ref{fig:32} that the approximate error between the exact and simulated outage probabilities of the MIMO-HARQ-CC scheme is acceptable. In addition, the asymptotic outage curves of MIMO-HARQ schemes are parallel to each other under high SNR. This observation is consistent with the results of the diversity order, which reflects the decreasing slope of the outage curve. Although the three MIMO-HARQ schemes have the same diversity order, the HARQ-IR scheme achieves the best outage performance among them because of its highest coding and modulation gain. \textcolor[rgb]{0.00,0.07,1.00}{Furthermore, by comparing to the MIMO-V2V communication scheme without using HARQ (labeled as ``No-HARQ'' in Figs. \ref{fig:22} - \ref{fig:32}), the proposed HARQ-assisted schemes exhibit a superior performance.}
\begin{figure}[!htb]
        \centering
        \includegraphics[width=3.5in]{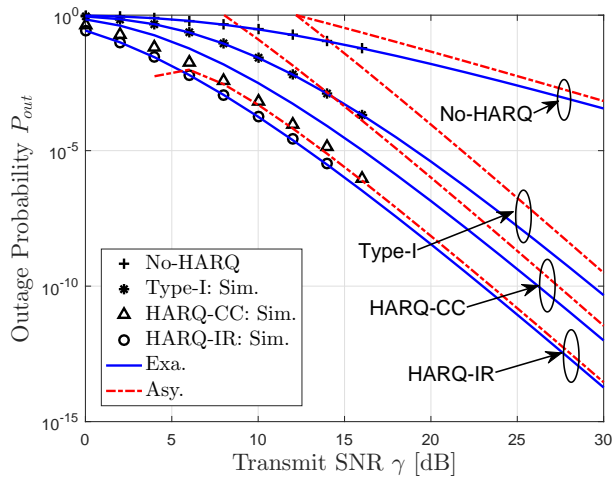}
        \caption{The outage probability $P_{out}$ versus the transmit SNR $\gamma$ with  $N_{T}=N_{R}=2$.}\label{fig:22}
\end{figure}
\begin{figure}[!htb]
    \centering
    \includegraphics[width=3.5in]{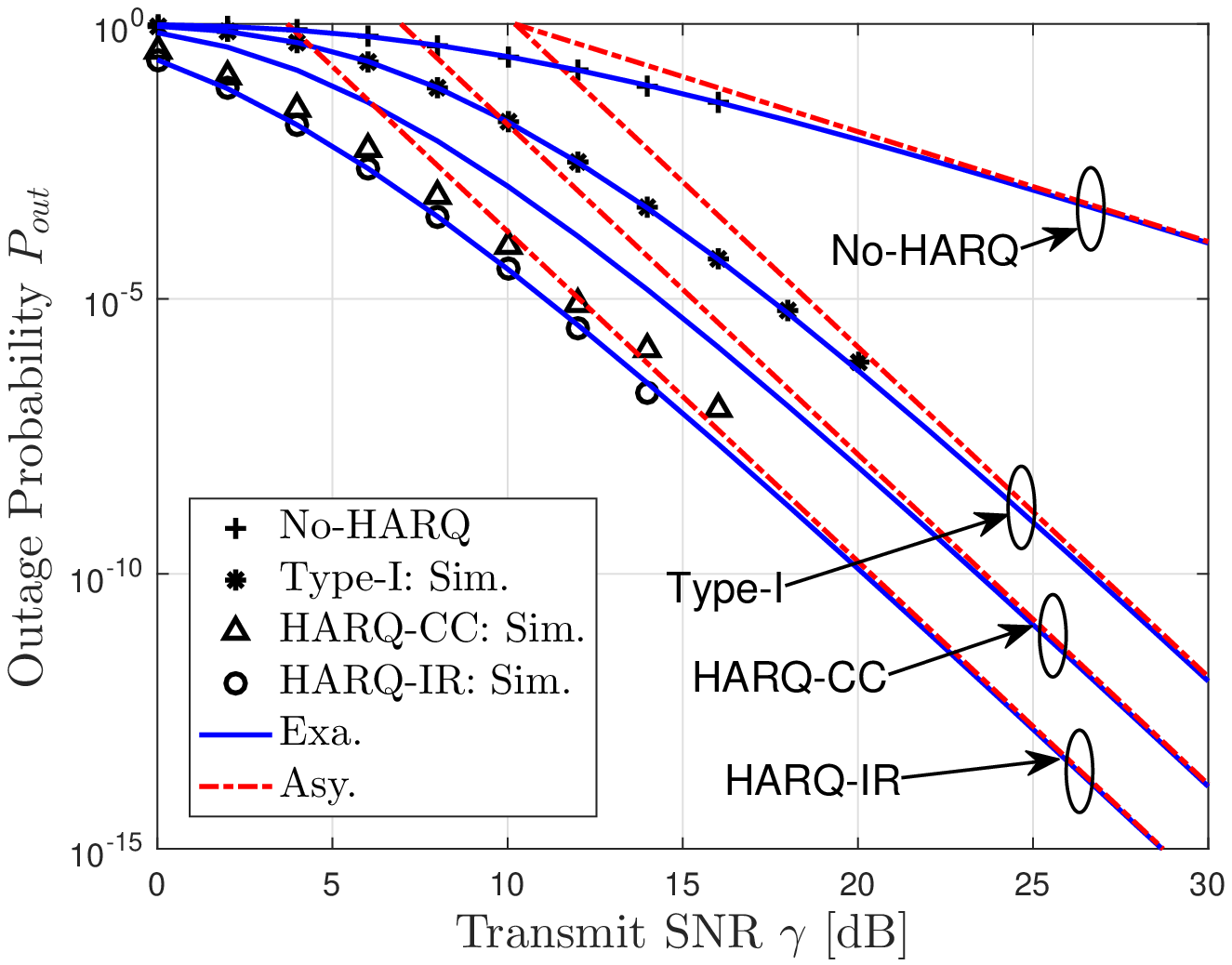}
    \caption{The outage probability $P_{out}$ versus the transmit SNR $\gamma$ with $N_{T}=3$ and $N_{R}=2$.}\label{fig:23}
\end{figure}
\begin{figure}[!htb]
    \centering
    \includegraphics[width=3.5in]{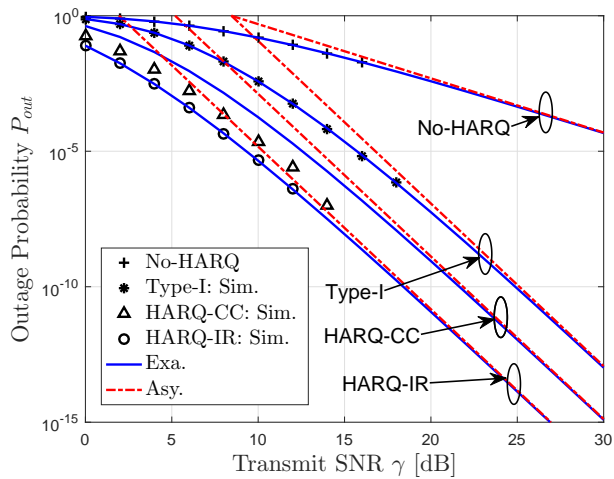}
    \caption{The outage probability $P_{out}$ versus the transmit SNR $\gamma$ with $N_{T}=2$ and $N_{R}=3$.}\label{fig:32}
\end{figure}

Figs. \ref{fig:R22} depicts the outage probability of the three MIMO-HARQ schemes versus the transmission rate $R$. As observed from Fig. \ref{fig:R22}, the exact outage curves coincide with the simulated outage ones except for the MIMO-HARQ-CC scheme, this is because the outage probability of the MIMO-HARQ-CC scheme is derived by using an upper bound. Besides, it can be seen that the outage probability is an increasing function of the transmission rate $R$, this is essentially due to the tradeoff between the throughput and reliability. Hence, the transmission rate should be properly chosen in practical V2V communication systems. Fortunately, the increasing monotonicity and convexity of the asymptotic outage probability with respect to $R$ can greatly ease the optimal rate selection. 
\begin{figure}
    \centering
    \includegraphics[width=3.5in]{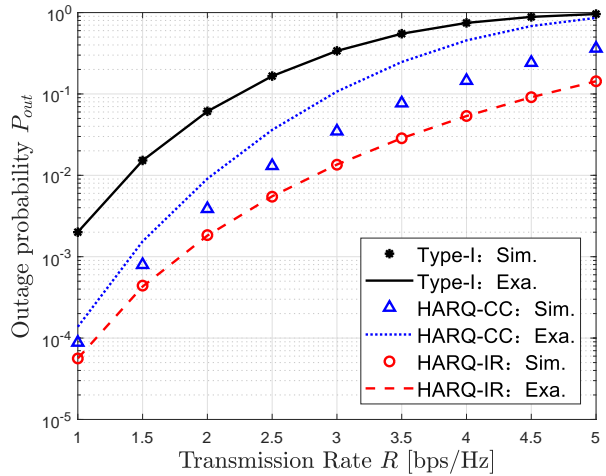}
    \caption{The outage probability $P_{out}$ versus the transmission rate $R$ by setting parameters as $\gamma=5$ dB.}\label{fig:R22}
\end{figure}

Additionally, Fig. \ref{fig:R23} shows the impacts of the transmission rate $R$  on the modulation and coding gain $\mathcal{C}(R)$ for the three MIMO-HARQ schemes. As expected, the modulation and coding gain decreases with the increase of the transmission rate. Besides, one can observe that the  HARQ-IR scheme obtains the maximum modulation and coding gain among the three MIMO-HARQ  schemes, and the MIMO-HARQ-CC scheme performs better than the MIMO-Type-I HARQ scheme. The observations in Fig. \ref{fig:R23}  are consistent with the asymptotic analysis in Section \ref{Coding_gain}.
\begin{figure}
    \centering
    \includegraphics[width=3.5in]{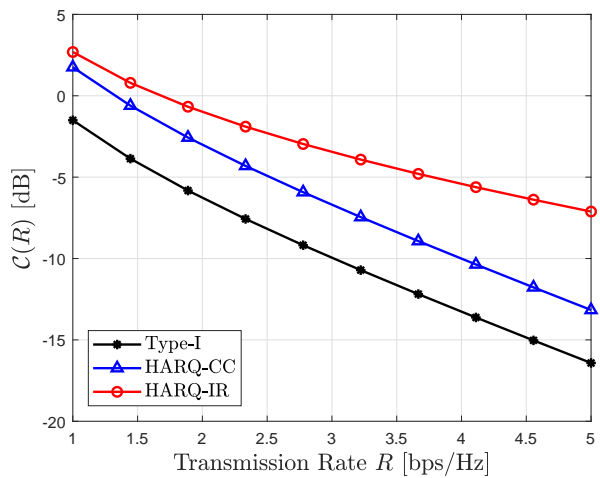}
    \caption{The modulation and coding gain $\mathcal{C}(R)$ versus the transmission rate $R$ by setting parameters as $N_{T}=N_{R}=2$ and $K=3$.}\label{fig:R23}
\end{figure}

To investigate the impact of the number of antennas on the outage performance, Fig. \ref{fig:NT} plots the outage probability against the number of transmit antennas $N_T$. It is shown that the increase of $N_T$ would improve the outage performance. Meanwhile, as the number of transmit antennas increases, there exists an outage floor for MIMO-HARQ schemes, which is only determined by the number of receive antennas. This result can be proved by applying the deterministic equivalent ${\left\| {{{\bf{v}}_k}} \right\|^2}/{N_T}\xrightarrow[{N_T \to \infty }]{a.s.}  1$ to \eqref{eqn:covsuff}, $k=1,\cdots,K$. Besides, it is found that the MIMO-Type-I HARQ scheme exhibits the worst outage performance among the three MIMO-HARQ schemes for $\gamma=5\, \mathrm{dB}$. By increasing the transmit SNR of the MIMO-Type-I HARQ scheme from $5$ dB to $11.5$ dB, it can reach almost the same performance as the MIMO-HARQ-IR scheme at $5$ dB. Although the MIMO-HARQ-IR scheme has the lowest outage probability, it relies on larger buffer size and higher computational complexity compared to the MIMO-HARQ-CC scheme. Furthermore, as the number of transmit antennas $N_T$ increases, one can observe from Fig. \ref{fig:NT} that the MIMO-HARQ-CC scheme achieves almost the same performance as the MIMO-HARQ-IR scheme. This result is consistent with our analysis in Section \ref{Large}. Therefore, in MIMO-HARQ systems with large-scale antenna arrays (i.e., massive MIMO), the MIMO-HARQ-CC scheme could be the best choice due to its lower coding complexity and hardware requirement without significantly degrading the outage performance. \textcolor[rgb]{0.00,0.07,1.00}{Additionally, to understand the physical meaning of the modulation and coding gain, the case of $N_T=N_R=2$, $K=3$ and $R=3$~bps/Hz is taken as an example. From Fig. \ref{fig:R23}, compared to the MIMO-Type-I HARQ scheme, the MIMO-HARQ-IR scheme improves the modulation and coding gain at $R=3$~bps/Hz by roughly 6.5~dB. This observation agrees well with Fig. \ref{fig:NT}, where the MIMO-HARQ-IR scheme reduces the required SNR by $11.5-5=6.5$~dB to ensure the same outage as the MIMO-Type-I HARQ scheme. Nonetheless, the obtained modulation and coding gain for the MIMO-HARQ-CC scheme  (i.e., \eqref{eqn:coding_gain}) is a lower bound for its actual value, because the result is established by relying on the upper bound of the outage probability.}
\begin{figure}
    \centering
    \includegraphics[width=3.5in]{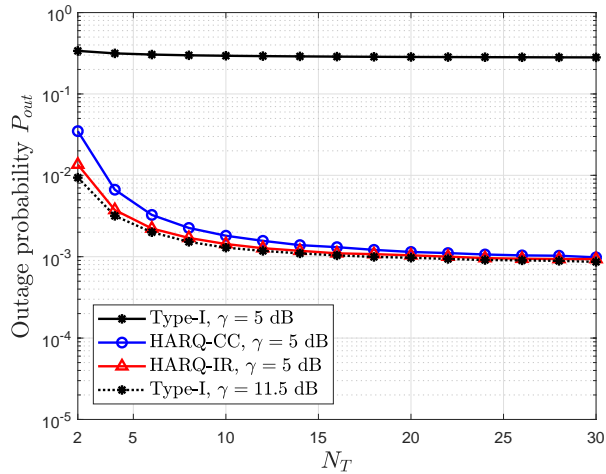}
    \caption{The outage probability $P_{out}$ versus the number of transmission antenna $N_{T}$ with $N_R=2$ and $K=3$. }\label{fig:NT}
\end{figure}
\section{Conclusions}\label{sec:con}
In this paper, we have investigated the outage performance for MIMO-HARQ assisted V2V communications with keyhole effect. To compensate the performance degradation caused by the rank-deficiency of keyhole effect, three types of HARQ, i.e., Type-I, HARQ-CC, and HARQ-IR, have been employed to boost the transmission reliability. With the help of MGF and Mellin transform, we have derived the exact outage expressions for the MIMO-Type-I and MIMO-HARQ-IR schemes, and an outage upper bound has been derived for the MIMO-Type-CC scheme. Moreover, the asymptotic outage analyses have been conducted to obtain meaningful insights. In particular, we have proved that the three MIMO-HARQ schemes can achieve the same diversity order. However, it has been shown that only full time diversity order can be achieved, while full spatial diversity order is unreachable as compared to MIMO-HARQ systems without keyhole effect. 
Furthermore, in large-scale antenna systems, we have proved that the MIMO-HARQ-CC scheme can achieve almost the same performance as the MIMO-HARQ-IR scheme. In the end, the Monte Carlo simulations have validated the analytical outcomes.
\appendices
\section{Proof of Theorem \ref{The:asy_I}}\label{eqn:8}
It is obviously that the asymptotic behavior of (\ref{eqn:op_type_I}) under high SNR, i.e., $\gamma_k\to\infty$, corresponds to the behavior of the PDF (\ref{eqn:pdf_X_l}) at small $x$, i.e., $x\to 0$. Based on \cite[Eq. (10.30.2), Eq. (10.30.3)]{lozier2003nist}, the modified Bessel function for $x\to 0$ is asymptotic to
 \begin{equation}\label{eqn:asyp_Kv}
{K_\tau}\left( x \right) \simeq \left\{ {\begin{array}{*{20}{c}}
{\frac{1}{2}\Gamma \left( \tau \right){{\left( {\frac{1}{2}x} \right)}^{ - \tau}},}&{\tau > 0},\\
{ - \ln \left( x \right)},&{\tau = 0}.
\end{array}} \right.
\end{equation}
 By substituting the asymptotic expression (\ref{eqn:asyp_Kv}) into (\ref{eqn:pdf_X_l}), the asymptotic outage probability for MIMO-Type-I HARQ scheme can be written as
 \begin{align}
P_{out}^{Type-I} \simeq &\left\{ {\begin{array}{*{20}{c}}
\prod\limits_{k = 1}^K {\int\limits_0^{\frac{{{N_T}}}{{{\gamma _k}}}\left( {{2^R} - 1} \right)} {\frac{{2{x^{\left( {{N_T} + {N_R}} \right)/2 - 1}}\frac{1}{2}\Gamma \left( {\tau} \right){{\left( {\frac{1}{2}2\sqrt x } \right)}^{ - \tau}}}}{{\Gamma \left( {{N_T}} \right)\Gamma \left( {{N_R}} \right)}}dx} }, &{\textup{$\tau > 0$}},\\
{\left( { - 1} \right)^K}\prod\limits_{k = 1}^K {\int\limits_0^{\frac{{{N_T}}}{{{\gamma _k}}}\left( {{2^R} - 1} \right)} {\frac{{2{x^{{N_T} - 1}}\ln \left( {2\sqrt x } \right)}}{{{\Gamma ^{\rm{2}}}\left( {{N_T}} \right)}}dx} },&{\textup{$\tau=0$}}.
\end{array}} \right.
\end{align}
Accordingly, the following two cases are treated individually.
\subsection{$\tau > 0$}
For the case of $\tau > 0$, we have
  \begin{align}
  P_{out}^{Type-I}
 &\simeq \prod\limits_{k = 1}^K {\frac{{\Gamma \left( {\tau} \right)}{\left( {\frac{{{N_T}}}{{{\gamma _k}}}\left( {{2^R} - 1} \right)} \right)}^{\left( {{N_T} + {N_R}} \right)/2 -\tau/2}}{{\Gamma \left( {{N_T}} \right)\Gamma \left( {{N_R}} \right)}\left({\left( {{N_T} + {N_R}} \right)/2 - \tau/2}\right)}}.
 \end{align}
\subsection{$\tau = 0$} For the case of $\tau=0$, one has
  \begin{align}\label{eqn:27}
P_{out}^{Type-I}
&\mathop  \simeq
{\left( { - 1} \right)^K}\prod\limits_{k = 1}^K \frac{{2{{\left( {\frac{{{N_T}}}{{{\gamma _k}}}\left( {{2^R} - 1} \right)} \right)}^{{N_T}}}}}{{{\Gamma ^{\rm{2}}}\left( {{N_T}} \right)}}\int\limits_0^1 {{x^{{N_T} - 1}}\ln \left( {2\sqrt {\frac{{{N_T}}}{{{\gamma _k}}}\left( {{2^R} - 1} \right)x} } \right)dx} \notag\\
&\mathop \simeq \limits^{\left( a \right)} \prod\limits_{k = 1}^K {\frac{{{{\left( {\frac{{{N_T}}}{{{\gamma _k}}}\left( {{2^R} - 1} \right)} \right)}^{{N_T}}}\ln \left( {{\gamma _k}} \right)}}{{{\Gamma ^{\rm{2}}}\left( {{N_T}} \right)}}\int\limits_0^1 {{x^{{N_T} - 1}}dx} } \notag\\
 &= \prod\limits_{k = 1}^K {\frac{{{{\left( {\frac{{{N_T}}}{{{\gamma _k}}}\left( {{2^R} - 1} \right)} \right)}^{{N_T}}}\ln \left( {{\gamma _k}} \right)}}{{{N_T}{\Gamma ^{\rm{2}}}\left( {{N_T}} \right)}}},
\end{align}
where step (a) holds by using $-2\ln \left( {2\sqrt {{{{N_T}}}/{{\gamma _k}}\left( {{2^R} - 1} \right)x} } \right)\simeq \ln\left( {{\gamma _k}} \right)$.

As a consequence, the asymptotic outage probability for the MIMO-Type-I HARQ scheme is derived as (\ref{eqn:asy I}).
\section{Proof of Theorem \ref{The:pcdf_Y_K}}\label{eqn:CC}
It is obviously found that $X_{CC}$ is the sum of independent random variables, which motivate us to use the method of MGF. More specifically, the MGF of the PDF of ${X_{CC}}= \sum\nolimits_{k = 1}^K {{\gamma _k}/{N_T}}X_k$ can be derived as a product of the MGFs of the PDFs of the random variables $X_1,\cdots, X_K$ as
 \begin{align}\label{eqn:mgf_Y_k}
{M_{X_{CC}}}\left( s \right) & = \prod\limits_{k = 1}^K {\int\limits_0^\infty  {{e^{\frac{{{\gamma _k}}}{{{N_T}}}xs}}{f_{{X_k}}}\left( x \right)dx} } \notag \\
 &= {\left( {\frac{2}{{\Gamma \left( {{N_T}} \right)\Gamma \left( {{N_R}} \right)}}} \right)^K}\prod\limits_{k = 1}^K {\int\limits_0^\infty  {{x^{\left( {{N_T} + {N_R}} \right)/2 - 1}}{e^{\frac{{{\gamma _k}}}{{{N_T}}}xs}}{K_{\tau}}\left( {2\sqrt x } \right)dx} }.
\end{align}
 By using \cite[Eq. (6.643.3), Eq. (9.222.2)]{jeffrey2007table}, 
 the MGF of the PDF of $X_{CC}$ can be expressed as
 \begin{align}
{M_{X_{CC}}}\left( s \right)
 &= \left\{ {\begin{array}{*{20}{c}}
{\prod\limits_{k = 1}^K {\frac{{\int\limits_0^\infty  {{t^{{N_T} - 1}}{e^{ - t}}{{\left( {1 - \frac{{{\gamma _k}}}{{{N_T}}}st} \right)}^{ - {N_R}}}dt} }}{{\Gamma \left( {{N_T}} \right)}},} }&{{N_T}  \ge{N_R}},\\
{\prod\limits_{k = 1}^K {\frac{{\int\limits_0^\infty  {{t^{{N_R} - 1}}{e^{ - t}}{{\left( {1 - \frac{{{\gamma _k}}}{{{N_T}}}st} \right)}^{ - {N_T}}}dt} }}{{\Gamma \left( {{N_R}} \right)}},} }&{{N_T} < {N_R}}.
\end{array}} \right.
\end{align}
With the help of \cite[Eq. (9.211.4)]{jeffrey2007table}, we can obtain the closed-form expression for the MGF of the PDF of $X_{CC}$ as
\begin{align}\label{eqn:mgf_yk_rew1}
{M_{X_{CC}}}\left(s \right) =\left\{ {\begin{array}{*{20}{c}}
\prod\limits_{k = 1}^K {{{\left( -{\frac{{{N_T}}}{{{\gamma _ks}}}} \right)}^{{N_T}}}}{\Psi \left( {{N_T},\tau+1;-\frac{{{N_T}}}{{{\gamma _k}}}{s^{ - 1}}} \right)},&{{N_T}  \ge {N_R}},\\
\prod\limits_{k = 1}^K {{{\left(- {\frac{{{N_T}}}{{{\gamma _k}s}}} \right)}^{{N_R}}}}{\Psi \left( {{N_R},\tau+1;-\frac{{{N_T}}}{{{\gamma _k}}}{s^{ - 1}}} \right)} ,&{{N_T} < {N_R}},
\end{array}} \right.
\end{align}

where $\Psi \left( \alpha,\gamma;z \right)$ denotes the confluent hypergeometric function \cite[eq. (9.211.4)]{jeffrey2007table}.

Finally, by applying the inverse Laplace transform, the CDF of $X_{CC}$ can be obtained as 
\begin{align}\label{eqn:F_y_k_cond_rew}
{F_{X_{CC}}}\left( x \right) & ={\mathcal L^{ - 1}}\left\{ {{M_{{X_{CC}}}}\left( { - s} \right)} \right\} = \frac{1}{{2\pi {\rm i}}}\int\limits_{a - {\rm i}\infty }^{a + {\rm i}\infty } {\frac{{{M_{{X_{CC}}}}\left( { - s} \right)}}{s}{e^{sx}}ds}.
\end{align}
By substituting \eqref{eqn:mgf_yk_rew1} into \eqref{eqn:F_y_k_cond_rew}, the CDF of $X_{CC}$ can be obtained as (\ref{eqn:cdf_Yk}). We thus complete the proof.

\section{Proof of Theorem \ref{lem:op-cc}}\label{app:op-CC}
By using \cite[Eq. (13.2.16), Eq. (13.2.19)]{lozier2003nist} 
and ignoring the higher order terms, the confluent hypergeometric functions in (\ref{eqn:mgf_yk_rew1}) are asymptotically equal to
\begin{equation}\label{eqn:confluen_hyp_asy}
\Psi \left( {{N_T},\tau+1;\frac{{{N_T}}}{{{\gamma _k}}}{s^{ - 1}}} \right) \simeq\left\{ {\begin{array}{*{20}{c}}
{\frac{{\Gamma \left(\tau \right)}}{{\Gamma \left( {{N_T}} \right)}}{{\left( {\frac{{{N_T}}}{{{\gamma _k}}}{s^{ - 1}}} \right)}^{ - \tau}},}&{{N_T} > {N_R}},\\
{\frac{1}{{\Gamma \left( {{N_T}} \right)}}\ln \left( {{\gamma _k}} \right),}&{{N_T} = {N_R}},\\
\end{array}} \right.
\end{equation}
\begin{equation}\label{eqn:confluen_hyp_asy2}
\Psi \left( {{N_R},\tau+1;\frac{{{N_T}}}{{{\gamma _k}}}{s^{ - 1}}} \right) \simeq {\begin{array}{*{20}{c}}{\frac{{\Gamma \left( \tau \right)}}{{\Gamma \left( {{N_R}} \right)}}{{\left( {\frac{{{N_T}}}{{{\gamma _k}}}{s^{ - 1}}} \right)}^{ - \tau}},}&{{N_T} < {N_R}}.
\end{array}}
\end{equation}

Then, by putting (\ref{eqn:confluen_hyp_asy}) and (\ref{eqn:confluen_hyp_asy2}) into (\ref{eqn:F_y_k_cond_rew}) along with the help of inverse Laplace transform, we can easily arrive at 
(\ref{eqn:as_cc}). 
\section{Proof of Theorem \ref{lem:pcdf_Z_K}}\label{eqn:IR}
The Mellin transform of the PDF of $X_{IR}$ can be written as a product of the Mellin transforms of the PDFs of  ${1 + {{{\gamma _1}}}/{{{N_T}}}X_1}, \cdots, {1 + {{{\gamma _K}}}/{{{N_T}}}X_K}$. With this property, the Mellin transform w.r.t. $X_{IR}$ can be rewritten as
\begin{align}\label{eqn:mellin_trans_Z_k}
\left\{ {{\cal M}{f_{X_{IR}}}} \right\}\left( s \right)&=\int\limits_0^\infty  {{x^{s - 1}}{f_{X_{IR}}}\left( x \right)dx}=\prod\limits_{k = 1}^K {\int\limits_0^\infty  {{{\left( {1 + \frac{{{\gamma _k}}}{{{N_T}}}x} \right)}^{s - 1}}{f_{{X_k}}}\left( x \right)dx} }\buildrel \Delta \over =  \varphi \left( s \right).
\end{align}
By putting (\ref{eqn:pdf_X_l}) into (\ref{eqn:mellin_trans_Z_k}) and using \cite[Eq. (7.811.5)]{jeffrey2007table}, we can obtain the Mellin transform of the PDF of $X_{IR}$ as
 \begin{align} \label{eqn:Mellin}
 \varphi \left( s \right)
 = &\prod\limits_{k = 1}^K \frac{{{{\left( {\frac{{{\gamma _k}}}{{{N_T}}}} \right)}^{s - 1}}}}{{\Gamma \left( {{N_T}} \right)\Gamma \left( {{N_R}} \right)}}\int\limits_0^\infty  {{x^{ - 1}}{{\left( {x + \frac{{{N_T}}}{{{\gamma _k}}}} \right)}^{s - 1}}G_{0,2}^{2,0}\left( {\left. {\begin{array}{*{20}{c}}
 - \\
{{N_T},{N_R}}
\end{array}} \right|x} \right)dx}  \notag \\
 =& \prod\limits_{k = 1}^K {\frac{{G_{1,3}^{3,1}\left( {\left. {\begin{array}{*{20}{c}}
1\\
{1 - s,{N_T},{N_R}}
\end{array}} \right|\frac{{{N_T}}}{{{\gamma _k}}}} \right)}}{{\Gamma \left( {{N_T}} \right)\Gamma \left( {{N_R}} \right)\Gamma \left( {1 - s} \right)}}}.
 \end{align}
 According to \cite[Eq. (6)]{yang2021asymptotic}, the CDF of $X_{IR}$ can be derived by invoking inverse Mellin transform as
 \begin{align}\label{eqn:cdf_IR_inverse}
{F_G}\left( x \right)&=  \left\{{\mathcal M} ^{ - 1}{\left[ - \frac{1}{s}\varphi \left( {s + 1} \right)\right]}\right\}\left( x \right)=\frac{{  1}}{{2\pi {\rm{i}}}}\int\nolimits_{c - {\rm i}\infty }^{c + {\rm i}\infty } {\frac{{{x^{ - s}}}}{-s}\varphi \left( {s + 1} \right)ds}.
\end{align}
By substituting (\ref{eqn:Mellin}) into (\ref{eqn:cdf_IR_inverse}), the CDF of $X_{IR}$ can finally be obtained as (\ref{eqn:CDF_Z_k_der}).

 \section{Proof of  Theorem \ref{The:IR_2}}\label{app:IR_2}
In order to derive the asymptotic expression of the MIMO-HARQ-IR scheme, by following \cite[Eq. 9.303]{jeffrey2007table}, the Meijer G-function of (\ref{eqn:CDF_Z_k_der}) can be expanded as
\begin{align}\label{eqn:meijer_g_expan}
&G_{1,3}^{3,1}\left( {\left. {\begin{array}{*{20}{c}}
1\\
{s,{N_T},{N_R}}
\end{array}} \right|\frac{{{N_T}}}{{{\gamma _k}}}} \right)\notag\\
&=\Gamma \left( {{N_T} - s} \right)\Gamma \left( {{N_R} - s} \right)\Gamma \left( {2 + s} \right){\left( {\frac{{{N_T}}}{{{\gamma _k}}}} \right)^{ s}}{}_1{F_2}\left( {s;1 + s - {N_T},1 + s - {N_R};{{\left( { - 1} \right)}^{ - 3}}\frac{{{N_T}}}{{{\gamma _k}}}} \right)\notag \\
&+ \Gamma \left( {s - {N_T}} \right)\Gamma \left( {{N_R} - {N_T}} \right)\Gamma \left( {{N_T}} \right){\left( {\frac{{{N_T}}}{{{\gamma _k}}}} \right)^{{N_T}}}{}_1{F_2}\left( {{N_T};{N_T} + 1 - s,1 + {N_T} - {N_R};{{\left( { - 1} \right)}^{ - 3}}\frac{{{N_T}}}{{{\gamma _k}}}} \right)\notag \\
&+ \Gamma \left( {s - {N_R}} \right)\Gamma \left( {{N_T} - {N_R}} \right)\Gamma \left( {{N_R}} \right){\left( {\frac{{{N_T}}}{{{\gamma _k}}}} \right)^{{N_R}}}{}_1{F_2}\left( {{N_R};1 + {N_R} - s,1 + {N_R} - {N_T};{{\left( { - 1} \right)}^{ - 3}}\frac{{{N_T}}}{{{\gamma _k}}}} \right),
\end{align}
where ${}_1F_2\left( \cdot;\cdot,\cdot;x \right)$ is a hypergeometric function  as \cite[Eq. 9.14.1]{jeffrey2007table}. By using the representation of the series expansion of the hypergeometric function and ignoring the higher order terms relative to $\gamma_k^{-1}$, the asymptotic expression of (\ref{eqn:meijer_g_expan}) for $N_T<N_R$ and $N_T>N_R$ can be obtained as
\begin{equation}\label{eqn:meijer_g_expan_asy}
G_{1,3}^{3,1}\left( {\left. {\begin{array}{*{20}{c}}
1\\
{s,{N_T},{N_R}}
\end{array}} \right|\frac{{{N_T}}}{{{\gamma _k}}}} \right) \simeq \left\{ {\begin{array}{*{20}{c}}
\Gamma \left( {s - {N_T}} \right)\Gamma \left( \tau \right)\Gamma \left( {{N_T}} \right){\left( {\frac{{{N_T}}}{{{\gamma _k}}}} \right)^{{N_T}}},&{N_T<N_R},\\
 \Gamma \left( {s - {N_R}} \right)\Gamma \left( \tau \right)\Gamma \left( {{N_R}} \right){\left( {\frac{{{N_T}}}{{{\gamma _k}}}} \right)^{{N_R}}},&{N_T>N_R}.
\end{array}} \right.
\end{equation}

Moreover, for the case of $N_T=N_R$, by utilizing the residue theorem, the Meijer G-function in (\ref{eqn:CDF_Z_k_der}) can be rewritten as
\begin{align}\label{eqn:meijer_G_asymp1}
G_{1,3}^{3,1}\left( {\left. {\begin{array}{*{20}{c}}
1\\{ s,{N_T},{N_R}}
\end{array}} \right|\frac{{{N_T}}}{{{\gamma _k}}}} \right)
=& \frac{1}{{2\pi {\rm i}}}\int\limits_{\mathcal L} {\Gamma \left( {s - t} \right){\Gamma ^2}\left( {{N_R} - t} \right)\Gamma \left( t \right){{\left( {\frac{{{N_T}}}{{{\gamma _k}}}} \right)}^t}dt} \notag \\
\simeq& - {\rm{Res}}\left\{ {\Gamma \left( {s - t} \right){\Gamma ^2}\left( {{N_R} - t} \right)\Gamma \left( t \right){{\left( {\frac{{{N_T}}}{{{\gamma _k}}}} \right)}^t}} , t = {N_R}\right\} \notag \\
=& {\left. { - \frac{d}{{dt}}\left( {{{\left( {t - {N_R}} \right)}^2}\Gamma \left( {s - t} \right){\Gamma ^2}\left( {{N_R} - t} \right)\Gamma \left( t \right){{\left( {\frac{{{N_T}}}{{{\gamma _k}}}} \right)}^t}} \right)} \right|_{t = {N_R}}}, \notag \\
 \end{align}
where the first step holds by using the definition of Meijer G-function \cite[Eq. (9.301)]{jeffrey2007table} and ignoring the higher order terms. Thus we can obtain the asymptotic expression for the case of $N_T=N_R$ as
 \begin{align}\label{eqn:meijer_G_asymp2}
  G_{1,3}^{3,1}\left( {\left. {\begin{array}{*{20}{c}}
1\\{ s,{N_T},{N_R}}
\end{array}} \right|\frac{{{N_T}}}{{{\gamma _k}}}} \right)
=&  - {\left. {\frac{d}{{dt}}\left( {\Gamma \left( { s - t} \right){\Gamma ^2}\left( {{N_R} - t + 1} \right)\Gamma \left( t \right)} \right)} \right|_{t = {N_R}}}{\left( {\frac{{{N_T}}}{{{\gamma _k}}}} \right)^{{N_R}}} \notag\\
 &-\Gamma \left( {s - {N_R}} \right){\Gamma ^2}\left( {{N_R} - {N_R} + 1} \right)\Gamma \left( {{N_R}} \right){\left( {\frac{{{N_T}}}{{{\gamma _k}}}} \right)^{{N_R}}}\ln \left( {\frac{{{N_T}}}{{{\gamma _k}}}} \right) \notag\\
\simeq  & \Gamma \left( {s - {N_R}} \right)\Gamma \left( {{N_R}} \right){\left( {\frac{{{N_T}}}{{{\gamma _k}}}} \right)^{{N_R}}}\ln {{{\gamma _k}}} .
\end{align}
Finally, by substituting (\ref{eqn:meijer_g_expan_asy}) and (\ref{eqn:meijer_G_asymp2}) into (\ref{eqn:CDF_Z_k_der}) together with \cite[Eq. (9.301)]{jeffrey2007table}, we can obtain the asymptotic outage probability for the MIMO-HARQ-IR scheme as (\ref{eqn:as_IR}).
\section{Proof of Theorem \ref{the:las}}\label{app:Le_2}
To prove \eqref{eqn:cov}, it suffices to show that
\begin{equation}\label{eqn:covsuff}
{\det \left( {{{\bf{I}}_{{N_T}}} + \sum\limits_{k = 1}^K {{\alpha _k}{{\bf{v}}_k}{{\bf{v}}_k}^{\rm{H}}} } \right)} \xrightarrow[{N_T \to \infty }]{a.s.} \prod\limits_{k = 1}^K {\left( {1 + {\alpha _k}{{\left\| {{{\bf{v}}_k}} \right\|}^2}} \right)},
\end{equation}
where $\xrightarrow{a.s.}$ denotes ``almost sure convergence'' and $\alpha_{k}={\gamma_k}/{N_T}{{\left\| {{{\bf{u}}_k}} \right\|}^2},~k\in\{1,\cdots,K\}$. To proceed with the proof, we reuse the definition ${{\bf{A}}_n} = {{\bf{I}}_{{N_T}}} + \sum\nolimits_{k = 1}^n {{\alpha _k}{{\bf{v}}_k}{{\bf{v}}_k}^{\rm{H}}}$, where $n\in [0,K]$ and we stipulate ${{\bf{A}}_0}={{\bf{I}}_{{N_T}}}$.
Then, by substituting the definition of ${{\bf{A}}_n}$ into the left hand side of \eqref{eqn:covsuff}, we arrive at
\begin{align}\label{eqn:det_mutu_inf_log_re1}
\det \left( {{\bf{A}}_{K }} \right) &=\det \left( {{{\bf{I}}_{N_T}} + \sum\limits_{k = 1}^K {{\alpha _k}{{\bf{v}}_k}{{\bf{v}}_k}^{\rm{H}}} } \right) \notag\\
 &= \det \left( {{{\bf{A}}_{K - 1}}} \right)\left( 1 + {\alpha _K}{{\bf{v}}_K}^{\rm{H}}{{\bf{A}}_{K - 1}}^{ - 1}{{\bf{v}}_K}\right)\notag\\
&=\det \left( {{{\bf{A}}_{K - 1}}} \right)\left( 1 + {\alpha _K}{{\bf{v}}_K}^{\rm{H}}{{\bf{A}}_{K - 2}}^{ - 1}{{\bf{v}}_K} - {\alpha _K}{\alpha _{K - 1}}\frac{{{{\left| {{{\bf{v}}_K}^{\rm{H}}{{\bf{A}}_{K - 2}}^{ - 1}{{\bf{v}}_{K - 1}}} \right|}^2}}}{{1 + {\alpha _{K - 1}}{{\bf{v}}_{K - 1}}^{\rm{H}}{{\bf{A}}_{K - 2}}^{ - 1}{{\bf{v}}_{K - 1}}}} \right),
\end{align}
where the first step holds by using $\mathrm{det}(\mathbf{I}+\mathbf{A}\mathbf{B})=\mathrm{det}(\mathbf{I}+\mathbf{B}\mathbf{A})$, and the second step holds by using the following recursive relationship between ${{\bf{A}}_{K-1}}$ and ${{\bf{A}}_{K-2}}$,
\begin{align}\label{eqn:A_k_inverse}
{\bf{A}}_{K - 1}^{ - 1}&= {\left( {{{\bf{A}}_{K - 2}} + {\alpha _{K - 1}}{{\bf{v}}_{K - 1}}{{\bf{v}}_{K - 1}}^{\rm{H}}} \right)^{ - 1}} \notag\\
& = {{\bf{A}}_{K - 2}}^{ - 1} - \frac{{{\alpha _{K - 1}}{{\bf{A}}_{K - 2}}^{ - 1}{{\bf{v}}_{K - 1}}{{\bf{v}}_{K - 1}}^{\rm{H}}{{\bf{A}}_{K - 2}}^{ - 1}}}{{1 + {\alpha _{K - 1}}{{\bf{v}}_{K - 1}}^{\rm{H}}{{\bf{A}}_{K - 2}}^{ - 1}{{\bf{v}}_{K - 1}}}},
\end{align}
and \eqref{eqn:A_k_inverse} is obtained by capitalizing on the Woodbury matrix identity. Repeatedly using the Woodbury matrix identity for ${\bf{A}}_{n}$ leads to
%
 \begin{align}\label{eqn:det_mutu_inf_log_re2}
&\det \left( {{{\bf{I}}_{N_T}} + \sum\limits_{k = 1}^K {{\alpha _k}{{\bf{v}}_k}{{\bf{v}}_k}^{\rm{H}}} } \right) \notag\\
&= \det \left( {{{\bf{A}}_{K - 1}}} \right)\left( 1 + {\alpha _K}{{\bf{v}}_K}^{\rm{H}}{{\bf{A}}_0}^{ - 1}{{\bf{v}}_K} - {\alpha _K}\sum\limits_{k = 1}^{K - 1} {{\alpha _k}\frac{{{{\left| {{{\bf{v}}_K}^{\rm{H}}{{\bf{A}}_{k - 1}}^{ - 1}{{\bf{v}}_k}} \right|}^2}}}{{1 + {\alpha _k}{{\bf{v}}_k}^{\rm{H}}{{\bf{A}}_{k - 1}}^{ - 1}{{\bf{v}}_k}}}}   \right)\notag\\
& = \det \left( {{{\bf{A}}_{K - 1}}} \right)\left( 1 + {\alpha _K}{{\left\| {{{\bf{v}}_K}} \right\|}^2}- {\alpha _K}\sum\limits_{k = 1}^{K - 1} {{\alpha _k}\frac{{{{\left| {{{\bf{v}}_K}^{\rm{H}}{{\bf{A}}_{k - 1}}^{ - 1}{{\bf{v}}_k}} \right|}^2}}}{{1 + {\alpha _k}{{\bf{v}}_k}^{\rm{H}}{{\bf{A}}_{k - 1}}^{ - 1}{{\bf{v}}_k}}}}  \right).
\end{align}
Since ${{\bf{A}}_1},\cdots, {{\bf{A}}_{K-1}}$ are positive definite matrices with bounded spectral norm, as $N_T\to\infty$, it follows by using \cite[Theorems 3.4, 3.7]{couillet2011random} that
\begin{equation}\label{eqn:det_eq}
\frac{1}{{{N_T}}}{{\bf{v}}_K}^{\rm{H}}{{\bf{A}}_{k - 1}}^{ - 1}{{\bf{v}}_k}\xrightarrow[{N_T \to \infty }]{a.s.} 0,\, k\ne K,
\end{equation}
\begin{equation}\label{eqn:det_eq1}
\frac{1}{{{N_T}}}{{\bf{v}}_k}^{\rm{H}}{{\bf{A}}_{k - 1}}^{ - 1}{{\bf{v}}_k}\xrightarrow[{N_T \to \infty }]{a.s.} \frac{1}{N_T}{\rm tr}({{\bf{A}}_{k - 1}}^{ - 1}).
\end{equation}
Plugging \eqref{eqn:det_eq} and \eqref{eqn:det_eq1} into \eqref{eqn:det_mutu_inf_log_re2} results in
\begin{align}\label{eqn:det_mutu_inf_log_re3}
\det \left( {{\bf{A}}_{K }} \right) &=\det \left( {{{\bf{I}}_{N_T}} + \sum\limits_{k = 1}^K {{\alpha _k}{{\bf{v}}_k}{{\bf{v}}_k}^{\rm{H}}} } \right)\notag\\
&\xrightarrow[{N_T \to \infty }]{a.s.} \det \left( {{{\bf{A}}_{K - 1}}} \right)\left( {1 + {\alpha _K}{{\left\| {{{\bf{v}}_K}} \right\|}^2}} \right).
\end{align}

By repeatedly applying the steps \eqref{eqn:det_mutu_inf_log_re1}-\eqref{eqn:det_mutu_inf_log_re3} to $\det \left( {{{\bf{A}}_{K - 1}}} \right)$, the deterministic equivalent of $\det \left( {{{\bf{A}}_{K - 1}}} \right)$ can be obtained as the recursive relationship in \eqref{eqn:det_mutu_inf_log_re3}. By recursively using \eqref{eqn:det_mutu_inf_log_re3}, we finally complete the proof.

\bibliographystyle{ieeetran}
\bibliography{manuscript_1}
\end{document}